\documentclass{article}
\usepackage{amsmath,amssymb,amsthm}
\usepackage{graphicx}
\usepackage{appendix}
\usepackage{fullpage}
\usepackage{parskip}
\usepackage{setspace}
\usepackage{multirow}
\usepackage{multicol}
\usepackage{bm}
\usepackage{natbib}
\usepackage{thmtools}
\usepackage{hyperref}
\hypersetup{
    colorlinks = true,
    citecolor = {blue},
    linkcolor = {blue}
}
\usepackage{enumitem}
\usepackage{color}
\usepackage[colorinlistoftodos]{todonotes}

\newcommand{\Bbf}{\mathbf{B}}
\newcommand{\E}{\mathbb{E}}
\newcommand{\V}{\mathbf{V}}
\newcommand{\U}{\mathbf{U}}
\newcommand{\Prb}{\mathbb{P}}

\newcommand{\R}{\mathbb{R}}
\newcommand{\N}{\mathbb{N}}
\newcommand{\Z}{\mathbb{Z}}

\newcommand{\sumin}{\sum_{i=1}^n}

\newcommand{\diag}{\text{diag}}
\newcommand{\diam}{\text{diam}}

\newtheorem{lem}{Lemma}
\newtheorem{theorem}{Theorem}
\newtheorem{assum}{Assumption}

\title{A Uniform Bound on the Operator Norm\\of Sub-Gaussian Random Matrices 
and Its Applications\footnote{We appreciate valuable comments and suggestions from Victor Chernozhukov, Guido 
Kuersteiner (Co-editor), three anonymous referees, and the participants of the 
conference 
on econometrics celebrating Peter Phillips' 40 years at Yale.}}

\date{\today}

\author{\setcounter{footnote}{2}
	Grigory Franguridi\footnote{
		Department of Economics, University of Southern California.
	}
\and
	Hyungsik Roger Moon\footnote{
		Department of Economics, University of Southern California and Yonsei University.
	}
}

\begin{document}
\maketitle

\begin{abstract}
    For an $N \times T$ random matrix $X(\beta)$ with weakly dependent uniformly sub-Gaussian entries $x_{it}(\beta)$ that may depend on a possibly infinite-dimensional parameter $\beta\in \mathbf{B}$, we obtain a uniform bound on its operator norm of the form $\E \sup_{\beta \in \mathbf{B}} ||X(\beta)|| \leq CK \left(\sqrt{\max(N,T)} + \gamma_2(\mathbf{B},d_\mathbf{B})\right)$, where $C$ is an absolute constant, $K$ controls the tail behavior of (the increments of) $x_{it}(\cdot)$, and $\gamma_2(\mathbf{B},d_\mathbf{B})$ is Talagrand's functional, a measure of multi-scale complexity of the metric space $(\Bbf,d_\Bbf)$. We illustrate how this result may be used for estimation that seeks to minimize the operator norm of moment conditions as well as for estimation of the maximal number of factors with functional data.
\end{abstract}


\onehalfspacing

\section{Introduction}

	Since its introduction in nuclear physics \citep{Wigner1955} and mathematical statistics \citep{Wishart1928}, random matrix theory has been developed to understand the properties of the spectra of large dimensional random matrices generated by various distributions. These include the asymptotic theory of the empirical distribution of the eigenvalues of large dimensional random matrices and bounds on the extreme eigenvalues. For detailed results on these topics, readers can refer to recent surveys like \citet{Bai2008}, \citet{EdelmanRao2005}, \citet{BaiSilverstein2010}, and \citet{Tao2012}, among others.

	In random matrix theory the study of the asymptotics of the largest eigenvalue of large dimensional random matrices goes back to \citet{Geman1980}. 	
	Suppose that $X$ is an $N \times T$ matrix consisting of random variables $x_{it}$. Many researchers have derived the limit of the largest eigenvalue of the sample covariance matrix, $\lambda_1(X'X)$\footnote{$X'$ denotes the transpose of $X$.}, under various distributional assumptions on the random matrix $X$. For example, when $X_{it}$ are iid $N(0,1)$ and $\kappa := \lim \frac{N}{T}$, \citet{Geman1980} showed that $\frac{1}{N} \lambda_1(X'X) \rightarrow_{a.s.} (1 + \kappa^{1/2})^2$. \citet{Johnstone2001} obtained a stronger result that the properly normalized largest eigenvalue, $\frac{\lambda_1(X'X) - \mu_{NT}}{\sigma_{NT}}$ with $\mu_{NT} = (\sqrt{N-1} + \sqrt{T} )^2$ and $\sigma_{NT} = (\sqrt{N-1} + \sqrt{T} )(1/\sqrt{N-1} + 1/\sqrt{T})^{1/3}$, converges to the Tracy--Widom law; this has been later shown to hold under more general distributional assumptions by \citet{khorunzhiy2012high} and \citet{tao2011random}, among many others.
	
	The aforementioned results imply that $\lambda_1(X'X)$ is stochastically bounded\footnote{A sequence of random variables $\xi_n$ is said to be \emph{stochastically bounded} or order $a_n$, $\xi_n=O_p(a_n)$, if for any $\varepsilon>0$ there exists $M>0$ such that $\Prb(|\xi_n/a_n|\geq M) \leq \varepsilon$ for all large enough values of $n$.} of order $\max(N,T)$, or equivalently, the operator norm $\| X \| := \sqrt{\lambda_1(X'X)}$ is stochastically bounded of order $\sqrt{\max(N,T)}$. In fact, such bound does not require that the underlying distribution is Gaussian and can be derived under much weaker conditions. For example, \citet{Latala2005} showed that the bound holds if $x_{it}$ are independent across $(i,t)$ with mean zero and uniformly bounded fourth moments. \citet{MoonWeidner2017} extended this result for the cases where $x_{it}$ are weakly correlated across $i$ or $t$. Other papers that have established similar bounds on $E\| X \| $ include \citet{Bandeira2016}, \citet{Guedonetal2017} and \citet{latala2018dimension}. 
	
	In the case where $X$ consists of independent sub-Gaussian entries, the $\sqrt{\max(N,T)}$ order for the operator norm may be obtained using a powerful way of bounding sub-Gaussian stochastic processes called \emph{generic chaining}, which was developed in \citet{fernique1975regularite} and advanced later by M. Talagrand in a series of papers. Indeed, note that $||X||=\max_{u\in\U}\max_{v\in\V} u'Xv$, where maxima are taken over the unit spheres $\U \subset \R^N$ and $\V\subset \R^T$, respectively. The process $Z(u,v)=u'Xv$ defined on $\U\times \V$ can be shown to be sub-Gaussian and so we can invoke generic chaining to get the bound for its expected maximum in terms of a certain measure of metric complexity of $\U\times \V$ called \emph{Talagrand's functional} $\gamma_2(\U\times\V)$ (see definition in the next section). It turns out that $\gamma_2(\U\times\V)$ has exact order $\sqrt{\max(N,T)}$.
	
	In this paper we extend existing nonasymptotic bounds on the operator norm of a high-dimensional random matrix to the case of elements that are allowed to be weakly dependent and to be functions of a possibly infinite-dimensional parameter. Specifically, let $x_{it}(\beta)$ be weakly dependent over $t$, sub-Gaussian stochastic processes indexed by parameter $\beta$ belonging to a (pseudo-)metric space $(\Bbf,d_\Bbf)$. Let $X(\beta)$ be the $N \times T$ matrix consisting of $x_{it}(\beta)$ and let $\gamma_2(\Bbf,d_\Bbf)$ be  Talagrand's functional of $\Bbf$ w.r.t. $d_{\Bbf}$ . Our main contribution is to show that $\E \sup_{\beta \in \mathbf{B}} \| X(\beta) \| $ is of order $\sqrt{\max(N,T)}+ \gamma_2(\Bbf,d_\Bbf)$.
	
We illustrate usefulness of this uniform bound with two examples. In one, we 
propose and show consistency of a new estimator that minimizes the operator 
norm of a matrix that consists of moment functions. In the other, we consider 
the generalization of the standard factor model to the case of functional data 
and suggest a new estimator of the maximal number of factors. 
	
The paper is organized as follows. Section \ref{section.theory} introduces our uniform bound along with the techniques necessary for its derivation. Section \ref{section.application} contains two applications of our theoretical result.  Finally, Section \ref{section.conclusion} concludes the paper. The appendix contains two technical proofs of the results in the main text. 

Throughout the paper, $C$ will denote a universal positive constant that may not be the same at each occurrence, but may never depend on sample sizes, dimensions or any other features of the modeling framework.

\section{Uniform bound on the operator norm}\label{section.theory}

\subsection{Generic chaining bound}

Our main result is based on the general bound on suprema of sub-Gaussian processes called the \emph{generic chaining bound}. We discuss this classic technique in this section and provide a proof in the appendix for completeness.

First, we need the following definitions. 
The \emph{$\psi$-Orlicz norm} of a random variable $Y$ is defined as
\begin{align*}
    ||Y||_{\psi} = \inf\left\{K>0 \text{ s.t. } \E\left(\psi(Y/K)\right) \leq 1\right\},
\end{align*}
where $\psi:\R_{+}\to \R_{+}$ is a convex function satisfying $\lim_{x\to \infty} \psi(x)/x = \infty$ and $\lim_{x\to 0} \psi(x)/x = 0$, and the convention that the infimum of an empty set is $+\infty$. In this paper, we let $\psi=\psi_2$, where $\psi_2(x)=\exp(x^2)-1$, and call $||\cdot||_{\psi_2}$ just ``the Orlicz norm''. A random variable with finite ($\psi_2$-)Orlicz norm is called \emph{sub-Gaussian}.

Intuitively, the Orlicz norm quantifies the decay speed for the tails of the 
distribution of $Y$. In fact, $||Y||_{\psi_2} \leq K$ is equivalent 
to\footnote{See e.g. \citet{vershynin2018high}, Proposition 2.5.2.}
\[
\Prb(|Y|\leq t) \geq 1-2e^{-t^2/K^2} \text{ for all } t\geq 0.
    \]
Hence, for example, Gaussian distributions and distributions with bounded support are all sub-Gaussian. 

Note also that the last inequality implies
\begin{align}\label{E:E|Y|<K}
\E |Y| = \int_0^\infty\left( 1-\Prb(|Y|\leq t)\right) dt \leq 2\int_0^\infty e^{-t^2/K^2}dt = K\sqrt{\pi}.    
\end{align}

Now let $T$ be a set and $d$ be a (pseudo-)metric on this set such that $(T,d)$ is a (pseudo-)metric space\footnote{Throughout the paper, ``metric'' can be replaced by a less restrictive notion of ``pseudometric'', a distinction we omit from now on.}. Consider a zero mean stochastic process $(Z_t)$ indexed by the elements of $T$. The process $(Z_t)$ is said to have \emph{sub-Gaussian increments} if there exists a constant $K>0$ such that
\begin{equation}\label{E:subgauss_increments}
||Z_t-Z_s||_{\psi_2} \leq K\cdot d(t,s) \text{ for all } t,s\in T.
\end{equation}

It has long been understood that behavior of sub-Gaussian processes is 
intimately connected to the metric complexity of its index set. In particular, 
the conventional bound on the expected supremum of $(Z_t)$ (see e.g. 
\citet{van1996weak} Corollary 2.2.8.) is
\begin{equation}\label{E:conventional_chaining}
\E \sup_{t\in T} Z_t \leq CK \int_0^\infty \sqrt{\log N(T,d,\varepsilon)}\, d\varepsilon,
\end{equation}
where $N(T,d,\varepsilon)$ is the \emph{covering number} of $(T,d)$ (i.e. the minimal number of $\varepsilon$-balls that is sufficient to cover $T$ in metric $d$) and $C$ is an absolute constant. The integral on the right hand side is sometimes called \emph{Dudley's entropy} of $(T,d)$ and quantifies complexity of $(T,d)$ across multiple scales.

It turns out, however, that Dudley's entropy bound is not optimal, even for 
Gaussian processes. In fact, the entropy may be infinite when the expected 
supremum is not, rendering the bound uninformative\footnote{For an illustrative 
example, see Exercise 8.1.12 in \citet{vershynin2018high}.}.


This led to the development of more precise ways to control suprema of 
sub-Gaussian processes in \citet{fernique1975regularite} and 
\citet{talagrand2006generic}. The \emph{generic chaining}  bound is stronger 
than \eqref{E:conventional_chaining} and is sharp for Gaussian 
processes\footnote{See Section 8.6 in \citet{vershynin2018high}.}. To introduce 
it, we need another definition.

For a metric space $(T,d)$, a sequence of finite subsets $T_0\subset T_1\subset \cdots \subset T$ is \emph{admissible} if their cardinalities satisfy
\begin{equation}
|T_0|=1 \text{ and } |T_k|\leq 2^{2^k} \text{ for } k \geq 1. \label{E:adm_seq}
\end{equation}
Let the distance from the point $t\in T$ to the set $T_k \subset T$ be
\[
d(t,T_k) = \inf_{t'\in T_k} d(t,t').
\]

\emph{Talagrand's functional} $\gamma_2$ is then defined by the formula
\begin{equation}
\gamma_2(T,d) = \inf_{(T_k)}\sup_{t\in T} \sum_{k=0}^\infty 2^{k/2}d(t,T_k), \label{E:Talagrand_func}
\end{equation}
where the infimum is taken over all admissible sequences $(T_k)$. Note that we can restrict our attention to only those admissible sequences that eventually come arbitrarily close to any point $t\in T$, which is possible provided $(T,d)$ is separable\footnote{A metric space $(T,d)$ is \emph{separable} if it has a countable subset that is dense in $T$.}.

To understand the relation between Talagrand's functional and Dudley's entropy, let us provide the discussion from \citet{talagrand2006generic} pp.12--13 here.

Denote $N_0=1$, $N_k=2^{2^k}$ for $k\geq 1$, and
\[
e_k(T)=\inf_{S \subset T: \,\, |S|\leq N_k} \sup_{t\in T} d(t,S).
\]
Note that
\begin{align}\label{E:gamma2_estimate_1}
\gamma_2(T,d) \leq  \inf_{(T_k)} \sum_{k=0}^\infty 2^{k/2} \sup_{t\in T} d(t,T_k) =  \sum_{k=0}^\infty 2^{k/2} e_k(T),
\end{align}
where the second equality holds because minimizing the sum $ \sum_{k=0}^\infty 2^{k/2} \sup_{t\in T} d(t,T_k)$ w.r.t. all admissible sequences $(T_k)$ can be performed by separately minimizing each term $\sup_{t\in T} d(t,T_k)$ w.r.t. subsets $T_k\subset T$ satisfying $|T_k|\leq N_k$. 

The definition of $e_k(T)$ involves choosing at most $N_k$ points $S$ in $T$ such that the balls with radius $e_k(T)$ and centers in $S$ cover $T$; moreover, $e_k(T)$ is the minimal such radius, i.e.
\[
e_k(T) = \inf\left\{\varepsilon>0:\,\, N(T,d,\varepsilon)\leq N_k\right\}.
\]
It follows that if $e_k(T)<\varepsilon$, then $N(T,d,\varepsilon)>N_k$ or $N(T,d,\varepsilon)\geq N_k+1$. Hence we can write
\begin{align*}
\sqrt{\log(N_k+1)}(e_{k}(T)-e_{k+1}(T)) \leq \int_{e_{k+1}(T)}^{e_{k}(T)} \sqrt{\log N(T,d,\varepsilon)} \,d\varepsilon.
\end{align*}
Since $\log(N_k+1)\geq 2^k \log 2$ for $k\geq 0$, summation over $k\geq 0$ yields
\begin{align}\label{E:entropy_estimate_1}
\sqrt{\log 2} \sum_{k=0}^\infty 2^{k/2}(e_{k}(T)-e_{k+1}(T)) \leq \int_{0}^{e_{0}(T)} \sqrt{\log N(T,d,\varepsilon)}\,d\varepsilon,
\end{align}
where, of course, $e_0(T)=\diam(T)=\sup_{t,s\in T} d(t,s)$.

The term on the left hand side of this inequality satisfies
\begin{align*}
\sum_{k=0}^\infty 2^{k/2}(e_{k}(T)-e_{k+1}(T)) = \sum_{k=0}^\infty 2^{k/2}e_{k}(T) - \sum_{k=1}^\infty 2^{(k-1)/2}e_{k}(T) \geq \left(1-2^{-1/2}\right) \sum_{k=0}^\infty 2^{k/2}e_{k}(T).
\end{align*}
Combining this with \eqref{E:gamma2_estimate_1} and \eqref{E:entropy_estimate_1} yields the key relation
\begin{align*}
\gamma_2(T,d) \leq C \int_{0}^{\diam(T)} \sqrt{\log N(T,d,\varepsilon)}\,d\varepsilon.
\end{align*}
Hence, when used as an upper bound, Talagrand's functional is sharper than Dudley's entropy.

We are now ready to state the generic chaining bound for sub-Gaussian processes, see e.g. Theorem 8.5.3 in \citet{vershynin2018high}.

\begin{theorem}[Generic chaining]\label{Thm:chaining}
Let $Z_t$, $t\in T$ be a mean zero random process on a separable metric space $(T,d)$ with sub-Gaussian increments as in \eqref{E:subgauss_increments}. Then, for some absolute constant $C>0$,
\[
\E \sup_{t\in T} Z_t \leq CK \gamma_2(T,d).
\]

\end{theorem}
\begin{proof}
See Appendix \ref{Appendix:proof_chaining}.
\end{proof}

\subsection{The main result}

We impose the following assumptions.

\begin{assum}\label{As:param_space}
The parameter $\beta$ belongs to a separable metric space $(\Bbf,d_\Bbf)$.
\end{assum}

\begin{assum}\label{As:MA}
For each $\beta\in\Bbf$, random variables $x_{it}(\beta)$ follow different MA($\infty$) processes for each $i$, viz.
\begin{align}\label{E:MA_infty}
x_{it}(\beta)=\sum_{\tau=0}^\infty \psi_{i\tau}(\beta) 
\varepsilon_{i,t-\tau}(\beta),
\end{align}
\end{assum}
where  $\psi_{i\tau}(\beta)$ are nonrandom coefficients such that, for all $i=1,\dots,N$ and $\beta\in\Bbf$,
\begin{align}\label{E:MA_infty_coeffs}
|\psi_{i\tau}(\beta)| \leq \theta_{\tau}, \text{ where } \sum_{\tau=0}^\infty \theta_{\tau} < \infty.
\end{align}

\begin{assum}\label{As:subgauss}
Innovations $\varepsilon_{i\tau}(\beta)$ are independent, mean zero, sub-Gaussian random variables with uniformly bounded scaling factors, i.e. there exists $K_1>0$ s.t. for all $(i,\tau,\beta)\in \N\times \Z \times \Bbf$
    \[
    ||\varepsilon_{i\tau}(\beta)||_{\psi_2} \leq K_1.
    \]
\end{assum}

\begin{assum}\label{As:subgauss_increments}
Innovations $\varepsilon_{i\tau}(\beta)$ are separable\footnote{Let $(\Bbf,d_\Bbf)$ be a separable metric space with a countable dense subset $D$. A stochastic process $\xi$ on $\Bbf$ is called \emph{separable} if for all $\beta\in \Bbf$, there exists a sequence $\beta_i\in D$ such that $\beta_i \to \beta$ and $\xi(\beta_i) \to \xi(\beta)$ almost surely. Non-separable stochastic processes have separable copies under very weak conditions, see \citet{shalizi2010book}.} stochastic processes whose increments are sub-Gaussian with uniformly bounded constants, i.e. there exists $K_2>0$ s.t. for all $(i,\tau)\in \N\times\Z$ and $(\beta_1,\beta_2)\in \Bbf \times \Bbf$
    \[
    ||\varepsilon_{i\tau}(\beta_1)-\varepsilon_{i\tau}(\beta_2)||_{\psi_2} \leq K_2 \cdot d_\Bbf(\beta_1,\beta_2).
    \]
\end{assum}

\autoref{As:param_space} is very weak and only imposes separability of the metric space $\Bbf$ which holds for most parameter spaces encountered in practice such as Euclidean spaces and spaces of integrable functions. \autoref{As:MA} is similar to case (ii) in Lemma S.2.1 of \citet{MoonWeidner2017} and allows $x_{it}(\beta)$ to be weakly dependent over time.
\autoref{As:subgauss} and \autoref{As:subgauss_increments} impose uniform sub-Gaussianity on the innovations $\varepsilon_{it}(\beta)$ and their increments $\varepsilon_{it}(\beta_1)-\varepsilon_{it}(\beta_2)$, respectively. Note that \autoref{As:subgauss_increments} is equivalent to the tail bound
\[
\Prb\left( |\varepsilon_{it}(\beta_1)-\varepsilon_{it}(\beta_2)| \leq t \cdot d_\Bbf(\beta_1,\beta_2) \right) \geq 1-2e^{-\frac{t^2}{K_2^2}} \text{ for all } t\geq 0.
\]

Denote $\psi_\tau(\beta) = (\psi_{1\tau}(\beta),\dots,\psi_{N\tau}(\beta))'$ and let $\Xi_{-\tau}(\beta)$ the $N\times T$ matrix consisting of $\varepsilon_{it}(\beta)$, $i=1,\dots,N$, $t=1-\tau,\dots,T-\tau.$ Equation \eqref{E:MA_infty} can be rewritten in the matrix form as
\begin{align*}
X(\beta) = (x_{ij}(\beta)) = \sum_{\tau=0}^\infty \diag(\psi_\tau(\beta)) \Xi_{-\tau}(\beta).
\end{align*}
Suppose for a moment that we have a bound on $\Xi_{-\tau}(\beta)$ of the form
\[
\E \sup_\beta ||\Xi_{-\tau}(\beta)|| \leq \varphi(N,T,\Bbf),
\]
where $\varphi$ does not depend on $\tau$. Then
\begin{align}
    \E \sup_\beta ||X(\beta)|| &= \E \sup_\beta \left\| \sum_{\tau=0}^\infty \diag(\psi_\tau(\beta)) \cdot \Xi_{-\tau}(\beta) \right\| \leq \E \sup_\beta \sum_{\tau=0}^\infty \left\|  \diag(\psi_\tau(\beta)) \cdot \Xi_{-\tau}(\beta) \right\| \notag \\
    &\leq \E \sum_{\tau=0}^\infty  \sup_\beta \left\|  \diag(\psi_\tau(\beta)) \cdot \Xi_{-\tau}(\beta) \right\| \leq \E \sum_{\tau=0}^\infty  \sup_\beta \left\|  \diag(\psi_\tau(\beta)) \right\|  \cdot\left\| \Xi_{-\tau}(\beta) \right\| \notag \\
    &\leq \E \sum_{\tau=0}^\infty  \sup_\beta \left\|  \diag(\psi_\tau(\beta)) \right\| \cdot \sup_\beta \left\| \Xi_{-\tau}(\beta) \right\| =  \sum_{\tau=0}^\infty \sup_\beta \left\|  \diag(\psi_\tau(\beta)) \right\| \E \sup_\beta \left\| \Xi_{-\tau}(\beta) \right\| \notag \\
	&\leq \varphi(N,T,\Bbf) \sum_{\tau=0}^\infty \sup_\beta \max_{i=1,\dots,N} |\psi_{i\tau}(\beta)| \leq D \varphi(N,T,\Bbf). \label{E:MAinfty_bound_same_as_MA0}
\end{align}
This shows that the bound on $\E \sup_\beta ||X(\beta)||$ is, up to the absolute constant $D$, the same as the bound on $\E \sup_\beta ||\Xi_{-\tau}(\beta)||$. Hence we can focus on obtaining the latter bound from now on. It will be clear from the proof that the bound will not depend on $\tau$, so we consider the case $\tau=0$ and denote $\Xi=\Xi_0$ for brevity.

The operator norm of $\Xi(\beta)$ can be expressed as 
\[
||\Xi(\beta)||=\sup_{u\in \U,v \in \V} Z(u,v,\beta),
\]
where $\U$ and $\V$ are unit spheres in $\R^N$ and $\R^T$, respectively, and the process 
\[
Z(u,v,\beta):=u'\Xi(\beta)v = \sum_{i=1}^N \sum_{t=1}^T u_i v_t \varepsilon_{it}(\beta).
\]
Define the $L_1$ product metric on $\U\times\V\times\Bbf$ by
\[
d((\tilde{u},\tilde{v}, \tilde{\beta}), (u,v,\beta)) = d_{\R^N}(\tilde{u},u) + d_{\R^T}(\tilde{v},v) + d_\Bbf(\tilde{\beta},\beta).
\]
where $d_{\R^d}$ denotes the standard Euclidean metric on $\R^d$.

To obtain a uniform bound on $||\Xi(\beta)||$, we would like to apply \autoref{Thm:chaining} to the process $Z(\cdot)$ defined on the metric space $(\U\times \V\times \Bbf, d)$. Our first lemma asserts that $Z$ has sub-Gaussian increments.

\begin{lem}\label{Lem:subgauss}
Under Assumptions \ref{As:param_space}, \ref{As:subgauss}, \ref{As:subgauss_increments}, the process $Z$ has sub-Gaussian increments w.r.t. the metric $d$, with the constant $K=\max(K_1,K_2).$
\end{lem}
\begin{proof}
For $(\tilde{u},\tilde{v}, \tilde{\beta}), (u,v,\beta) \in \U\times \V\times \Bbf$, write
\[
Z(\tilde{u},\tilde{v}, \tilde{\beta})-Z(u,v,\beta) = (\tilde{u}-u)'\Xi(\tilde{\beta})\tilde{v} + u'(\Xi(\tilde{\beta})-\Xi(\beta))\tilde{v} + u'\Xi(\beta) (\tilde{v}-v) = z_I + z_{II}+z_{III}.
\]
Recall a standard result for the $\psi_2$ norm (see e.g. equation (2.1) in \citet{mendelson2008subgaussian}): there exists an absolute constant $c>0$ such that for all constants $a_i$ and independent centered variables $\xi_1,\dots,\xi_n$ one has
\[
\left|\left|\sumin a_i \xi_i \right|\right|_{\psi_2} \leq c \sqrt{\sumin a_i^2 ||\xi_i||_{\psi_2}^2} \leq c ||a|| \max_{i=1,\dots,n} ||\xi_i||_{\psi_2}.
\]
Applying this inequality, we obtain
\begin{align*}
    ||z_I||_{\psi_2} &= \left|\left|\sum_{i=1}^N \sum_{t=1}^T(\tilde{u}_i-u_i)v_t \varepsilon_{it}(\tilde{\beta})\right|\right|_{\psi_2} \leq c K_1 \cdot d_{\R^N}(\tilde{u},u) ,\\
    ||z_{II}||_{\psi_2}  &= \left|\left|\sum_{i=1}^N \sum_{t=1}^T u_i v_t (\varepsilon_{it}(\tilde{\beta})-x_{it}(\beta)) \right|\right|_{\psi_2} \leq c K_2 \cdot d_\Bbf(\tilde{\beta},\beta),\\
    ||z_{III}||_{\psi_2} &= \left|\left|\sum_{i=1}^N \sum_{t=1}^T u_i (v_t-\tilde{v}_t) \varepsilon_{it}(\tilde{\beta})\right|\right|_{\psi_2} \leq c K_1 \cdot d_{\R^T}(\tilde{v},v).
\end{align*}
This implies
\begin{align*}
    ||Z(\tilde{u},\tilde{v}, \tilde{\beta})-Z(u,v,\beta)||_{\psi_2} &\leq ||z_{I}||_{\psi_2} + ||z_{II}||_{\psi_2} + ||z_{III}||_{\psi_2} \\
&\leq c\max(K_1,K_2)\cdot \left(d_{\R^N}(\tilde{u},u) + d_{\R^T}(\tilde{v},v) + d_\Bbf(\tilde{\beta},\beta)\right)\\
&= c\max(K_1,K_2)\cdot d((\tilde{u},\tilde{v}, \tilde{\beta}),(u,v,\beta)),
\end{align*}
which completes the proof.
\end{proof}

Our second lemma establishes the bound on Talagrand's functional of a product space in terms of Talagrand's functionals of component spaces.

\begin{lem}[Talagrand's functional of a product space] \label{Lem:Talagrand_prod}
Consider a finite number of metric spaces $(T_l,d_{l}),$ $l=1,\dots,L$ and the product space $T=\bigotimes_{l=1}^L T_l$ with the $L^1$ product metric defined by
\[
d(t ,t') = \sum_{l=1}^L d_{l}(t_l,t_l') \text{ for } t=(t_1,\dots,t_L), t'=(t_1',\dots,t_L') \in T.
\]
Talagrand's functional $\gamma_2$ of $T$ satisfies
\begin{align*}
    &\gamma_2(T,d) \leq (1+\sqrt{2}) \sum_{l=1}^L \gamma_2(T_l,d_l).
\end{align*}
\end{lem}
\begin{proof}
See Appendix \ref{Appendix:proof_Talagrand}.
\end{proof}


Finally, by \autoref{Lem:subgauss}, we can apply the generic chaining bound of \autoref{Thm:chaining} to $Z(u,v,\beta)$ defined on the separable metric space $T=\U\times \V \times \Bbf$ with the $L_1$ metric $d$. \autoref{Lem:Talagrand_prod} then yields
\begin{align}\label{E:sharpest_bnd}
\E \sup_{\beta \in \Bbf} ||\Xi(\beta)||  = \E\left[ \sup_{(u,v,\beta)\in \U\times\V\times \Bbf} Z(u,v,\beta)\right] \leq CK\left( \gamma_2(\U,d_{\R^N}) + \gamma_2(\V,d_{\R^T}) + \gamma_2(\Bbf,d_{\Bbf}) \right).
\end{align}

For the unit sphere $S^{d-1}$ in $\R^d$, its Dudley's entropy satisfies
\[
\int_0^{\diam(S^{d-1})} \sqrt{\log N(S^{d-1},||\cdot||,\varepsilon)} \,d\varepsilon \leq C \sqrt{d}.
\]
Besides, Talagrand's functional is bounded from above by Dudley's entropy (e.g. Exercise 8.5.7 in \citet{vershynin2018high}), up to absolute constant factors.

Applying these bounds to unit spheres $\U\subset \R^N$ and $\V\subset \R^T$ gives
\[
\E \sup_{\beta \in \Bbf} ||\Xi(\beta)|| \leq CK\left( \sqrt{\max(N,T)} + \gamma_2(\Bbf,d_\Bbf) \right).
\]

Finally, taking into account the inequality \eqref{E:MAinfty_bound_same_as_MA0}, we obtain the main theoretical result of this paper.

\begin{theorem}\label{Thm:main_result}
Under Assumptions \ref{As:param_space}, \ref{As:MA}, \ref{As:subgauss}, \ref{As:subgauss_increments},
\begin{align*}
\E \sup_{\beta \in \Bbf} ||X(\beta)|| \leq CK\left( \sqrt{\max(N,T)} + \gamma_2(\Bbf,d_\Bbf) \right),
\end{align*}
where $K=\max(K_1,K_2).$
\end{theorem}

\noindent {\bf Remarks}
\begin{itemize}
\item[(i)] Generic chaining yields not only the bound on the expected value, but also tail bounds and bounds on moments of $\sup_{\beta \in \Bbf} ||X(\beta)|| $, see e.g. \citet{dirksen2015tail}. In particular, it follows from Theorem 8.5.5 of \citet{vershynin2018high} that, for all $u\geq 0$, the event
\[
\sup_{\beta \in \Bbf} ||X(\beta)|| \leq CK\left[ \sqrt{\max(N,T)}+\gamma_2(\Bbf,d_\Bbf) + (2+\text{diam}(\Bbf))u\right]
\]
holds with probability at least $1-2\exp(-u^2)$, where $\text{diam}(\Bbf)$ is the diameter of $\Bbf$ in $d_\Bbf$.

\item[(ii)] Suppose $\varepsilon_{it}(\beta)$ are Gaussian random variables. Then the process $Z(u,v,\beta)$ is Gaussian and therefore the bound \eqref{E:sharpest_bnd} is sharp, up to an absolute constant, by the majorizing measure theorem, see Theorem 8.6.1 in \citet{vershynin2018high}.

\item[(iii)] If $\Bbf$ is a bounded set in $\R^d$, the main result and majorization of Talagrand's functional with Dudley's entropy yield
\[
\E \sup_{\beta \in \Bbf} ||X(\beta)|| \leq CK \sqrt{\max(N,T,d)}.
\]
In particular, if $\Bbf$ consists of one element (so that there is no dependence on $\beta$), the bound reduces to
\begin{align*}
\E ||X|| \leq CK \sqrt{\max(N,T)},
\end{align*}
which is a classical result in random matrix theory, see e.g. \citet{Latala2005}.

\item[(iv)] The dimension of $\Bbf$ is allowed to grow with the sample size; of course, to maintain the $\sqrt{\max(N,T)}$ rate for the operator norm, the dimension should not grow faster than $\sqrt{\max(N,T)}$.

\item[(v)]  \autoref{Thm:main_result} can be generalized to the case of Orlicz norms $||\cdot||_{\psi_\alpha}$ with $\psi_\alpha(x)=\exp(x^\alpha)-1$, $\alpha\geq 1$. An important special case $\alpha=1$ corresponds to \emph{sub-exponential} random variables.

The bound will take the form
\[
\E \sup_{\beta \in \Bbf} ||X(\beta)|| \leq CK\left( \left(\max(N,T)\right)^{1/\alpha} + \gamma_\alpha(\Bbf,d_\Bbf) \right),
\]
where the generalized Talagrand's functional is defined by
\[
\gamma_\alpha(T,d) = \inf_{(T_k)}\sup_{t\in T} \sum_{k=0}^\infty 2^{k/\alpha}d(t,T_k).
\]
The proof is similar to the case $\alpha=2$. The appropriate version of the generic chaining bound is
\[
\E\sup_{t\in T} Z(t) \leq CK \gamma_\alpha(T,d),
\]
where $Z(\cdot)$ is a stochastic process with bounded $\psi_\alpha$-Orlicz increments. Also,
\[
\gamma_\alpha(T,d) \leq C \int_0^{\diam(T)} \psi_\alpha^{-1}\left(N(T,d,\varepsilon)\right)\,d\varepsilon.
\]
Both results can be found in \citet{talagrand2006generic}.
\end{itemize}

\section{Applications}\label{section.application}

\subsection{Operator norm minimizing estimator}


In this section, we investigate a new estimator that minimizes the operator norm of the moment function matrix. Suppose that $\varepsilon_{it}(\beta) \in \mathbb{R}^L $ are $L$ moment functions of $\beta \in \mathbf{B} \subset \mathbb{R}^K$ such that $\E(\varepsilon_{it}(\beta_0)) = 0$. For simplicity, assume that $L = K = 1$.  Let $\varepsilon(\beta) = [\varepsilon_{it}(\beta)]$, the $N \times T$ matrix of moment functions.  

The conventional method of moment estimator solves
\[
\tilde{\beta} = \arg\min_{\beta \in \mathbf{B}} \left| \frac{1}{NT} \sum_{i,t} \varepsilon_{it}(\beta) \right| 
= \arg\min_{\beta \in \mathbf{B}} \left| \frac{\mathbf{1}_N^{\prime}}{\sqrt{N}} \left(\frac{ \varepsilon(\beta)}{\sqrt{NT}} \right) \frac{\mathbf{1}_T}{\sqrt{T}} \right|, 
\]
where $\mathbf{1}_N$ is the $N$-vector of ones. 

The new estimator we propose minimizes the operator norm of the moment function matrix $\varepsilon(\beta)$, 
\begin{align*}
\widehat{\beta} 
&:= \arg\min_{\beta \in \mathbf{B}} \frac{\| \varepsilon(\beta) \|}{\sqrt{NT}} \\
&= \arg\min_{\beta \in \mathbf{B}}  \sup_{\| w \| = 1, \| v \| = 1} w' \left( \frac{\varepsilon(\beta)}{\sqrt{NT}} \right) v.
\end{align*}
In this section we establish consistency of $\widehat{\beta}$ using our main result of the previous section.

\begin{assum}\label{as:LONE} (i) the parameter set $\mathbf{B}$ is a bounded subset of $\R$, (ii) the centered moment function $\varepsilon_{it}(\beta) - \E( \varepsilon_{it}(\beta))$ satisfies the conditions of Assumptions \ref{As:MA}-\ref{As:subgauss_increments}, and (iii) for any $\epsilon > 0$, there exists $\delta > 0$ such that $ \inf_{ | \beta - \beta_0 | \geq \epsilon}  \frac{\| \E(\epsilon(\beta)) \|}{\sqrt{NT}} \geq 2\delta$. 	
\end{assum}

Conditions (i)-(ii) of Assumption \ref{as:LONE} ensure that $\varepsilon_{it}(\beta) - \E( \varepsilon_{it}(\beta))$ satisfies Assumptions \ref{As:param_space}-\ref{As:subgauss_increments}. The last condition (iii) corresponds to the identification condition of the extremum estimator.

For consistency of $\widehat{\beta}$, it suffices to show that for any $\epsilon>0$, there exists $\delta > 0$ such that 
\begin{equation}
\inf_{ | \beta - \beta_0 | \geq \epsilon} \frac{\| \varepsilon(\beta) \|}{\sqrt{NT}} - \frac{\| \varepsilon(\beta_0) \|}{\sqrt{NT}} \geq \delta
\end{equation}
with probability approaching one.

First, note that, since $\E( \varepsilon(\beta_0))=0$, the triangle inequality yields
\begin{align}\label{E:eps_beta0}
 \frac{\| \varepsilon(\beta_0) \|}{\sqrt{NT}} \leq \sup_{\beta \in \mathbf{B}} \frac{\| \varepsilon(\beta) - \E(\varepsilon(\beta) \|}{\sqrt{NT}}.
\end{align}
On the other hand,
\begin{align}\label{E:eps_beta}
	\inf_{ | \beta - \beta_0 | \geq \epsilon} \frac{\| \varepsilon(\beta) \|}{\sqrt{NT}} 
	&\geq \inf_{ | \beta - \beta_0 | \geq \epsilon} \frac{\| \E( \varepsilon(\beta) ) \|}{\sqrt{NT}} - \sup_{\beta \in \mathbf{B}} \frac{\| \varepsilon(\beta) - \E(\varepsilon(\beta) \|}{\sqrt{NT}}. 
\end{align}
Combine \eqref{E:eps_beta0} and \eqref{E:eps_beta} to obtain
\begin{align}\label{E:eps_beta_beta0}
	\inf_{ | \beta - \beta_0 | \geq \epsilon} \frac{\| \varepsilon(\beta) \|}{\sqrt{NT}} - \frac{\| \varepsilon(\beta_0) \|}{\sqrt{NT}}
	&\geq \inf_{ | \beta - \beta_0 | \geq \epsilon} \frac{\| \E( \varepsilon(\beta) ) \|}{\sqrt{NT}} - 2\sup_{\beta \in \mathbf{B}} \frac{\| \varepsilon(\beta) - \E(\varepsilon(\beta) \|}{\sqrt{NT}}.
\end{align}
Finally, choose $\delta$ as in Assumption \ref{as:LONE}(iii) to guarantee
\[
	\inf_{ | \beta - \beta_0 | \geq \epsilon}  \frac{\| \E(\varepsilon(\beta)) \|}{\sqrt{NT}} \geq 2\delta
\]
and note that \autoref{Thm:main_result} gives
\[
\sup_{\beta \in \mathbf{B}} \frac{\| \varepsilon(\beta) - \E(\varepsilon(\beta) \|}{\sqrt{NT}} = o_p(1).
\]
Then \eqref{E:eps_beta_beta0} implies
\begin{align*}
\inf_{ | \beta - \beta_0 | \geq \epsilon} \frac{\| \varepsilon(\beta) \|}{\sqrt{NT}} - \frac{\| \varepsilon(\beta_0) \|}{\sqrt{NT}} \geq 2\delta - 2o_p(1) \geq \delta \quad \text{w.p.a.}1,
\end{align*}
which finishes the proof of consistency of $\hat{\beta}$.

\noindent {\bf Remarks}
\begin{itemize}
	\item[(i)] If $\varepsilon_{it}(\beta)$ are iid, then the identification condition Assumption \ref{as:LONE} (iii) becomes the usual identification condition, that is, for any $\epsilon > 0$, there exists $\delta > 0$ such that $ \inf_{ | \beta - \beta_0 | \geq \epsilon}  | \E( \varepsilon_{it}(\beta) ) | > 2\delta$. This is because $\frac{\| \E( \varepsilon(\beta) ) \|}{\sqrt{NT}} = |\E( \varepsilon_{it}(\beta) )| \left\| \frac{\mathbf{1}_N}{\sqrt{N}} \frac{\mathbf{1}_T'}{\sqrt{T}} \right\| = |\E( \varepsilon_{it}(\beta) )|.$
	
	\item[(ii)] Suppose that $\varepsilon_{it}(\beta) = (\varepsilon_{1,it}(\beta),...,\varepsilon_{L,it}(\beta))' \in \mathbb{R}^{L}$. Instead of the  operator norm objective function, we may also consider  
	\[
	\sum_{l=1}^L \omega_l \frac{ \| \varepsilon_l(\beta) \| }{\sqrt{NT}}, 
	\]
	where $\omega_l$ are weights.
	
	\item[(iii)] We can also extend the objective function to be the sum of $R_{NT}$ largest singular values, where $R_{NT}$ is a sequence of positive integers such that $R_{NT} \rightarrow \infty$ while $\frac{R_{NT}}{\sqrt{\min(N,T)}} \rightarrow 0$:
	\[
	\frac{1}{\sqrt{NT}}\sum_{r=1}^{R_{NT}} s_r(\varepsilon(\beta)),
	\] 
	where $s_r(A)$ is the $r^{th}$ largest singular value of matrix $A$.
	Since $\| \varepsilon(\beta) - \E(\varepsilon(\beta)) \| = s_1( \varepsilon(\beta) - \E(\varepsilon(\beta)) )$, we have
	\begin{align*}
		& \sup_{\beta \in \mathbf{B}} \frac{1}{\sqrt{NT}} \sum_{r=1}^{R_{NT}} s_r(\varepsilon(\beta) - \E(\varepsilon(\beta)) ) \\
		& \qquad \leq R_{NT} \frac{\| \varepsilon(\beta) - \E(\varepsilon(\beta)) \|}{\sqrt{NT}} 
		= O_p\left( \frac{R_{NT}}{\sqrt{\min(N,T)}} \right) = o_p(1).
	\end{align*}
\end{itemize}

\subsection{Estimator of number of factors with functional data}

Consider a generic factor model for functional data
\begin{align}\label{E:factor_model}
Y(\beta) = \lambda(\beta)f(\beta)'+ U(\beta),
\end{align}
where $\beta$ belongs to a separable metric space $(\Bbf,d_\Bbf)$, $Y(\beta)\in 
\R^{N\times T}$ is the observation matrix of functional outcomes $\beta \mapsto 
y_{it}(\beta)$, and $\lambda(\beta)\in \R^{N\times R(\beta)}$, $f(\beta)\in 
\R^{ T \times R(\beta)}$ such that for all $\beta\in \Bbf$ the probability 
limits of $\lambda(\beta)'\lambda(\beta)/N$ and $f(\beta)'f(\beta)/T$ exist and 
are positive definite deterministic matrices such that 
\begin{equation} 
	\liminf_{N,T}\sup_\beta 
	s_R\left(\frac{\lambda(\beta)f(\beta)'}{\sqrt{NT}}\right)  > 0. 
	\label{eq.ex.factor.1} 
\end{equation}
The object of interest is the maximal rank $R=\max_{\beta\in \Bbf}R(\beta)$.

To illustrate applicability of this model, suppose that the outcome variable is  intraday pollution levels $y_{it}(\beta)$, where $\beta$ is the time within a day, across counties $i$ and time $t$, as in \citet{aue2015prediction}. It is plausible to assume that counties with higher population density and dependence on automobiles will have higher average levels of pollution. At the same time, pollution patterns on weekdays and on weekends may differ in a systematic way. Hence it is reasonable to model the intraday pollution curve $y_{it}(\cdot)$ as the interaction of the county fixed effect $\lambda_i(\cdot)$ and the time effect $f_t(\cdot)$, plus independent noise, arriving at model \eqref{E:factor_model}. A related approach to modeling functional time series can be found in \citet{kargin2008curve}, whose empirical objective is to predict the contract rate curves of daily Eurodollar futures. 

Of course, arguments similar to those outlined above may be applied to modeling of numerous other functional quantities, from mortality as a function of age to crop yields as a function of spatial location. For more examples and an overview of functional data analysis, see e.g. \citet{wang2016functional} and \citet{kowal2019functional}.

Let us now show heuristically how to derive a consistent estimator of the maximal rank $R$.

Note that the model assumptions imply
\begin{align}
\sup_\beta s_i(\lambda(\beta)f(\beta)') &= O_p\left(\sqrt{NT}\right), \quad i\leq R,\\
\sup_\beta s_i(\lambda(\beta)f(\beta)') &= 0, \quad i> R.
\end{align}
If $U(\beta)$ satisfies the conditions of \autoref{Thm:main_result}, we have $\sup_\beta ||U(\beta)|| =O_p\left( \sqrt{\max(N,T)} + \gamma_2(\Bbf,d_\Bbf)\right)$ and so
\[
\sup_\beta ||U(\beta)|| = O_p\left(\sqrt{\max(N,T)}\right).
\]
Denote $s_i(A)$ the $i$-th largest singular value of matrix $A$. The Ky Fan inequality for singular values asserts that for $A,B \in \R^{N \times T}$
\[
|s_i(A+B)-s_i(A)| \leq s_1(B)=||B|| \text{ for all } i=1,\dots,\min(N,T).
\]

Using this inequality, for a fixed $\beta$ we obtain
\begin{align*}
s_R\left(\frac{Y(\beta)}{\sqrt{NT}}\right)= s_R\left(\frac{\lambda(\beta)f(\beta)'+U(\beta)}{\sqrt{NT}}\right) \geq s_R\left(\frac{\lambda(\beta)f(\beta)'}{\sqrt{NT}}\right) - \left|\left|\frac{U(\beta)}{\sqrt{NT}}\right|\right|.
\end{align*}
Therefore, there exists a positive constant $C>0$ such that
\begin{align*}
\sup_\beta s_R\left(\frac{Y(\beta)}{\sqrt{NT}}\right) &\geq \sup_\beta s_R\left(\frac{\lambda(\beta)f(\beta)'}{\sqrt{NT}}\right) - \sup_\beta \left|\left|\frac{U(\beta)}{\sqrt{NT}}\right|\right| \\
& \geq C - O_p \left(\frac{1}{\sqrt{\min(N,T)}}\right),
\end{align*}
where the last inequality holds by (\ref{eq.ex.factor.1}).

On the other hand,
\begin{align*}
\sup_\beta s_{R+1}\left(\frac{Y(\beta)}{\sqrt{NT}}\right) \leq \sup_\beta s_{R+1}\left(\frac{\lambda(\beta)f(\beta)'}{\sqrt{NT}}\right) + \sup_\beta s_{R+1}\left(\frac{U(\beta)}{\sqrt{NT}}\right) \leq O_p \left(\frac{1}{\sqrt{\min(N,T)}}\right).
\end{align*}
This establishes consistency of the following natural estimator of $R$,
\begin{align}
\hat{R} = \sum_{l=1}^{\min(N,T)} \mathbb{I}\left( \sup_\beta s_l\left( \frac{Y(\beta)}{\sqrt{NT}} \right) \geq \psi_{NT}\right), \label{E:max_rank_estimator}
\end{align}
where $\psi_{NT}$ is a sequence of real numbers satisfying $\psi_{NT}\to 0$ and $\psi_{NT} \sqrt{\min(N,T)} \to \infty$.

Empirical practice calls for an automatic procedure for choosing the tuning 
parameter $\psi_{NT}$. One may consider one of the following three options, 
using the penalty term from \citet{bai2002determining}:
\begin{align*}
\psi_{NT,1} &= \hat{\sigma} \sqrt{\frac{N+T}{NT} \log \frac{NT}{N+T}},\\
\psi_{NT,2} &= \hat{\sigma} \sqrt{\frac{N+T}{NT} \log \min(N,T)},\\
\psi_{NT,3} &= \hat{\sigma} \sqrt{ \frac{\log \min(N,T)}{\min(N,T)}},
\end{align*}
where $\hat{\sigma}^2=\sup_\beta \hat{\sigma}^2(\beta)$ is a consistent estimator of 
\[
\sigma^2=\sup_\beta \sigma^2(\beta) = \sup_\beta \frac{1}{NT}\sum_{i=1}^N \sum_{t=1}^T \E\left[ u_{it}(\beta)^2  \right].
\]
In applications, $\hat{\sigma}^2(\beta)$ can be replaced by the residual variance of $Y(\beta)$ after partialling out $k_{\max}$ factors using principle component analysis, where $k_{\max}$ is a pre-specified upper bound on the true maximal number of factors $R$.

\section{Monte Carlo illustration}\label{section.mc}
Here we illustrate the performance of the maximal rank estimator in the 
functional factor model described in the previous section with a simple 
simulation design.
	
The data generating process is the functional factor model \eqref{E:factor_model}, where, for simplicity, we let the loadings $\lambda(\beta)$ and the factors $f(\beta)$ to be independent of $\beta$. In scalar form, the model is
\begin{equation}
y_{it}(\beta)=\sum_{r=1}^{R(\beta)} \lambda_{ir} f_{tr} + u_{it}(\beta), \,\, i=1,\dots,N, \,t=1,\dots,T,
\end{equation}
where $\lambda_{ir}, f_{tr} \sim \text{iid} \, N(0,1)$ and 
\begin{equation*}
	u_{it}(\beta) = \frac{\sigma}{2} \left(\xi_{it1} \cos\beta + \xi_{it2} \sin\beta\right),\quad \xi_{it1}, \xi_{it2} \sim \text{iid} \, N(0,1).
\end{equation*}

The chosen specification for $u_{it}(\cdot)$ comes from a generic representation of any Gaussian stochastic process as an infinite trigonometric series, in which we only retain one term. Clearly, the error variance $\V(u_{it}(\beta))=\sigma$ for all $\beta$ and there is nontrivial dependence of $u_{it}(\beta)$ across values of $\beta$. We set $\sigma=1$. The results do not change substantially when larger values of $\sigma$ are used.

We choose the range of parameter $\beta$ to be $\Bbf = \{0,0.1,\dots,0.9,1\}$ and the corresponding ranks
\[
(R(0),R(0.1),\dots,R(0.9), R(1)) = (4,4,1,4,3,1,2,3,4,4,1),
\]
so that the true value of interest is $R=\max_\beta R(\beta)=4$.
	
	The simulated bias and root MSE for the maximal rank estimator \eqref{E:max_rank_estimator} are shown in \autoref{T:FFM}. Clearly, the choice $\psi_{NT}=\psi_{NT,3}$ for the tuning parameter leads to poor small sample performance, which is similar to the results of \citet{bai2002determining}. However, under the other two choices $\psi_{NT,1}, \psi_{NT,2}$, bias and RMSE are modest even in small samples and become essentially zero when $N,T \ge 50$.
	
	Given these simulation results, we are convinced that our generalization \eqref{E:max_rank_estimator} of the estimator of \citet{bai2002determining} will be useful for practitioners who are interested in estimating factor models with functional data.
	
	\begin{table}[]
\begin{center}
\begin{tabular}{lllllllllllll}
\multicolumn{1}{l}{}  &      & \multicolumn{3}{c}{\multirow{-1.5}{*}{$\psi_{NT,1}$}}                                  &  & \multicolumn{3}{c}{\multirow{-1.5}{*}{$\psi_{NT,2}$}}                                  &  & \multicolumn{3}{c}{\multirow{-1.5}{*}{$\psi_{NT,3}$}}                                  \\
\hline
$N$ & $T$    & 25                          & 50                          & 100                         &  & 25                          & 50                          & 100                         &  & 25                          & 50                          & 100                         \\
\hline
& \text{Bias} & 3.5 & 1.7 & 0.2 &  & 2.0 & 0.7 & 0.0 &  & 6.5 & 4.2 & 1.3 \\
\multirow{-2}{*}{25}  & \text{RMSE} & 0.6 & 0.6 &0.4 &  &0.6 & 0.6 & 0.1 &  & 0.5 & 0.6 & 0.6 \\
\hline
& \text{Bias} & 2.0 & 0.0 & 0.0 &  & 0.9 & 0.0 & 0.0 &  & 4.5 & 4.8 & 0.2 \\
\multirow{-2}{*}{50}  & \text{RMSE} & 0.6 & 0.1 & 0.0 &  & 0.6 & 0.0 & 0.0 &  & 0.6 & 0.6 & 0.4 \\
\hline
& \text{Bias} & 0.3 & 0.0 & 0.0 &  & 0.1 & 0.0 & 0.0 &  & 1.7 & 0.3 & 0.9 \\
\multirow{-2}{*}{100} & \text{RMSE} &0.5 & 0.0 & 0.0 &  & 0.2 & 0.0 & 0.0 &  & 0.6 & 0.5 & 0.5
\end{tabular}
\caption{Performance of the maximal rank estimator under different thresholds $\psi_{NT}$.}
\label{T:FFM}
\end{center}
\end{table}

\section{Conclusion}\label{section.conclusion}

In this paper, we derive a novel uniform stochastic bound on the operator norm of sub-Gaussian random matrices. We use it to establish consistency of a new estimator that minimizes the operator norm of the matrix of moment functions as well as to introduce an estimator of the maximal number of factors in a functional interactive fixed effects model.

\newpage

\appendix
\appendixpage

\begin{subappendices}

\section{Proof of \autoref{Thm:chaining}}\label{Appendix:proof_chaining}

The following proof can be found in \citet{vershynin2018high}, see Theorem 8.5.3.

Since $T$ is separable, we can assume for simplicity that it is finite. Let $(T_k)$ be an admissible sequence and $\pi_k(t)$ be the best approximation to $t$ in $T_k$, i.e.
\[
d(t,\pi_k(t)) = \min_{t'\in T_k} d(t,t').
\]

Now consider a chain of approximations to the point $t$ starting from some $t_0$
\[
t_0=\pi_0(t) \to \pi_1(t) \to \dots \to \pi_{\tilde{K}}(t)=t
\]
and write
\[
Z_t - Z_{t_0} = \sum_{k=1}^{\tilde{K}} \left(Z_{\pi_k(t)}- Z_{\pi_{k-1}(t)}\right).
\]
Sub-Gaussianity of increments $Z_{\pi_k(t)}- Z_{\pi_{k-1}(t)}$ implies that, for any $u> 0$,
\begin{align}
    \Prb\left(|Z_{\pi_k(t)}- Z_{\pi_{k-1}(t)} | \leq Cu 2^{k/2} d(\pi_k(t),\pi_{k-1}(t)) \right) \geq 1-2e^{-C^2 u^2 2^k/K^2} \geq 1-2e^{-8u^2 2^k}, \label{E:bnd1}
\end{align}
where $C \geq \sqrt{8K}$.

Note that since $\pi_k(t)\in T_k$, $\pi_{k-1}(t)\in T_{k-1}$, the number of possible pairs $(\pi_k(t),\pi_{k-1}(t))$ is $|T_k|\cdot |T_{k-1}| \leq |T_k|^2=2^{2^{k+1}}$. Applying the union bound to \eqref{E:bnd1} over $k\in \mathbb{N}$ and pairs $(\pi_k(t),\pi_{k-1}(t))$, we obtain
\begin{align}
\Prb\left(|Z_{\pi_k(t)}- Z_{\pi_{k-1}(t)} | \leq Cu 2^{k/2} d(\pi_k(t),\pi_{k-1}(t)) \text{ for all } t\in T, k\in \mathbb{N} \right) &\geq 1-\sum_{k=1}^\infty 2^{2^{k+1}}\cdot 2e^{-8u^2 2^k} \\
&\geq 1-2e^{-u^2}.
\end{align}
The event on the left-hand side implies
\begin{align*}
    |Z_t-Z_{t_0}| &\leq Cu \sum_{k=1}^\infty 2^{k/2} d(\pi_k(t),\pi_{k-1}(t)) \leq Cu \sum_{k=1}^\infty 2^{k/2} (d(\pi_k(t),t)+d(\pi_{k-1}(t),t)) \\
    &\leq \tilde{C}u \gamma_2(T,d)
\end{align*}
for a constant $\tilde{C}>0$. Taking supremum over $t\in T$ yields
\[
\sup_{t\in T} |Z_t-Z_{t_0}|\leq \tilde{C}u \gamma_2(T,d).
\]
Since this event holds with probability at least $1-2e^{-u^2}$, $\sup_{t\in T} |Z_t-Z_{t_0}|$ is a sub-Gaussian random variable with Orlicz norm bounded by $\tilde{C}\gamma_2(T,d)$. The conclusion then follows from \eqref{E:E|Y|<K} and the inequality
\[
\E\sup_{t\in T} Z_t = \E\sup_{t\in T} (Z_t-Z_{t_0}) \leq \E\sup_{t\in T} 
|Z_t-Z_{t_0}|. \,\, \blacksquare
\]

\section{Proof of \autoref{Lem:Talagrand_prod}}\label{Appendix:proof_Talagrand}

We give the proof for the case of $L=2$ metric spaces for simplicity. The case of arbitrary $L$ follows immediately by inspection.

Denote the two spaces by $(X,d_X)$ and $(Y,d_Y)$. Consider admissible sequences $(X_k)$ and $(Y_k)$ in $X$ and $Y$, respectively. To each such pair there corresponds a sequence $(\Tilde{T}_k)$ in $T=X\times Y$ of the form
\begin{align}
    \tilde{T}_k=
    \begin{cases}
    X_0\times Y_0, \quad k=0,\\
    X_{k-1}\times Y_{k-1}, \,\,k\geq 1. \label{E:T_tilde}
    \end{cases}
\end{align}
This sequence is admissible since $|\tilde{T}_0|=|\tilde{T}_1|=1$ and $|\tilde{T}_k|=|X_{k-1}||Y_{k-1}|\leq 2^{2^{k-1}}2^{2^{k-1}}=2^{2^k}$ for $k  \geq 2$.

Fix $(x,y)\in T$ and write
\[
\sum_{k \geq 0} 2^{k/2} d((x,y),\tilde{T}_k) = d_X(x,X_0) + \sum_{k\geq 1} 2^{k/2}d(x,X_{k-1}) + d_Y(y,Y_0) + \sum_{k\geq 1} 2^{k/2}d_Y(y,Y_{k-1}).
\]
The bound on the first two terms on the right-hand side is
\begin{align*}
    d_X(x,X_0) + \sum_{k\geq 1} 2^{k/2}d(x,X_{k-1}) &= (1+\sqrt{2})d_X(x,X_0) + \sqrt{2} \sum_{k \geq 1} 2^{k/2}d_X(x,X_k)\\
    &\leq (1+\sqrt{2}) \sum_{k \geq 0} 2^{k/2}d_X(x,X_k).
\end{align*}
Similarly, we have
\begin{align*}
    d_Y(y,Y_0) + \sum_{k\geq 1} 2^{k/2}d_Y(y,Y_{k-1}) \leq (1+\sqrt{2}) \sum_{k \geq 0} 2^{k/2}d_Y(y,Y_k).
\end{align*}
Adding the two inequalities and taking suprema yields
\begin{align*}
    \sup_{(x,y)} \sum_{k \geq 0} 2^{k/2} d((x,y),\tilde{T}_k) &\leq (1+\sqrt{2})\sup_{(x,y)} \left( \sum_{k \geq 0} 2^{k/2}d_X(x,X_k) + \sum_{k \geq 0} 2^{k/2}d_Y(y,Y_k) \right)\\
    &\leq (1+\sqrt{2})\left( \sup_x \sum_{k \geq 0} 2^{k/2}d_X(x,X_k) + \sup_y \sum_{k \geq 0} 2^{k/2}d_Y(y,Y_k) \right).
\end{align*}
Taking infima over admissible sequences $(\tilde{T}_k)$ (which are functions of admissible sequences $(X_k)$ and $(Y_k)$) yields
\begin{align*}
    \inf_{(\tilde{T}_k)} \sup_{(x,y)} \sum_{k \geq 0} 2^{k/2} d((x,y),\tilde{T}_k)&\leq (1+\sqrt{2})\left(\inf_{(X_k)}  \sup_x \sum_{k \geq 0} 2^{k/2}d_X(x,X_k) + \inf_{(X_k)} \sup_y \sum_{k \geq 0} 2^{k/2}d_Y(y,Y_k) \right)\\
    &= (1+\sqrt{2}) \cdot \left( \gamma_2(X,d_X) + \gamma_2(Y,d_Y)\right).
\end{align*}
Finally, note that $\gamma_2(T,d)$ is not larger than the left-hand side of the 
inequality above since the infimum in its definition is taken over all 
admissible sequences $(T_k)$, not only those that have the form 
$(\tilde{T}_k)$. $\blacksquare$
\end{subappendices}

\bibliographystyle{apalike}
\bibliography{references}

\end{document}